\documentclass[11pt,arabic,envcountsame]{llncs}
\let\doendproof\endproof
\renewcommand\endproof{~\hfill\qed\doendproof}
\usepackage{verbatim}
\usepackage{amssymb}
\usepackage[english,ruled]{algorithm2e}
\usepackage{amsmath}
\usepackage{url}
\usepackage{subfigure}
\usepackage{setspace}
\usepackage{xspace}
\usepackage{mdwlist}
\usepackage{textcomp}

\newcommand{\nbart}{NBART\xspace}

\newcommand{\mblock}{\textsc{Block}}
\newcommand{\mhash}{\textsc{Summary}}

\newcommand{\mreport}{\textsc{Report}}

\newcommand{\prods}{{\cal P}}
\newcommand{\cons}{{\cal C}}
\newcommand{\byz}{{\cal F}}
\newcommand{\byzp}{{\cal F}_\prods}
\newcommand{\byzc}{{\cal F}_\cons}
\newcommand{\nbyzp}{F_\prods}
\newcommand{\nbyzc}{F_\cons}
\newcommand{\nprods}{N_\prods}
\newcommand{\ncons}{N_\cons}
\newcommand{\nblocks}{B}

\newcommand{\mprodset}{\mbox{\emph{prodset}}}
\newcommand{\mconset}{\mbox{\emph{conset}}}
\newcommand{\vsigma}{\vec{\sigma}}
\newcommand{\expbenef}{\bar{\beta}}
\newcommand{\exputil}{\bar{u}}
\newcommand{\expcost}{\bar{\alpha}}
\newcommand{\pbenef}{\beta_\prods}
\newcommand{\cbenef}{\beta_\cons}
\newcommand{\playerset}{{\cal M}}

\newcommand{\idset}{{\cal I}}

\newcommand{\collset}{{\cal T}}

\newcommand{\nconids}{N_\collset}
\newcommand{\nconidsp}{N_\collset^\prods}
\newcommand{\nconidsc}{N_\collset^\cons}

\begin{document}
\title{Asynchrony and Collusion in the \\ N-party BAR Transfer Problem}

\author{
Xavier Vila\c{c}a \and
Oksana Denysyuk \and
Lu\'{\i}s Rodrigues
}

\institute{INESC-ID, Instituto Superior T\'{e}cnico, Universidade T\'{e}cnica de Lisboa}

\date{\today}

\maketitle

\begin{abstract}
  The problem of reliably transferring data from a set of $\nprods$
  producers to a set of $\ncons$ consumers in the BAR model, named
  N-party BAR Transfer (\nbart), is an important building block for
  volunteer computing systems.  An algorithm to solve this problem in
  synchronous systems, which provides a Nash equilibrium, has been
  presented in previous work.  In this paper, we propose an \nbart
  algorithm for asynchronous systems. Furthermore, we also address the
  possibility of collusion among the Rational processes. Our game
  theoretic analysis shows that the proposed algorithm tolerates
  certain degree of arbitrary collusion, while still fulfilling the
  \nbart properties.
 \end{abstract}

\section{Introduction}
Peer-to-peer networks can be used for executing
computationally intensive projects, as shown by the Boinc
infrastructure\,\cite{boinc}. Building systems on this kind
of networks may be quite challenging due to the existence of Byzantine processes, whose behaviour
is arbitrary, and of Rational processes, which may deviate from the specified
protocols if they can increase their utility. A system model that captures
this variety of behaviours has been coined the BAR model\,\cite{barb},
named after the three 
classes of processes (Byzantine, Altruistic, and Rational) that it explicitly considers.

Our work focuses on the particular problem of reliably transferring data from
a set of $\nprods$ producers to a set of $\ncons$ consumers in the BAR model, named
N-party BAR Transfer (\nbart). This problem is
an important building block for volunteer computing systems, since it allows
volunteers to transfer intermediate or final results of the computations to another
set of volunteers, after storing the data for some time. For instance, 
if computations are to be performed using a model such as MapReduce, mappers may invoke
the \nbart primitive to transfer the intermediate results to reducers.

Although an algorithm that solves  this problem has already been devised for  
synchronous systems\,\cite{Vilaca:11}, in a peer-to-peer network
it is often unrealistic to assume that there is a known upper bound
for the execution time and the communication delay. With this in mind, 
this paper addresses the \nbart problem in an asynchronous system. 

Furthermore, this paper also addresses the problem of collusion, which
is a real issue in peer-to-peer networks due to attacks, such as sybil
and white washing. In addition to arbitrary collusion of Byzantine
players, we consider that Rational processes may create collusion groups,
including producers and consumers.

\subsubsection{Related Work}
\label{sec:relatedwork}

Since models based on traditional Game Theory assume that all
processes follow the selfish strategy that maximises their utility
function, they fail to account for arbitrary behaviour that may arise
from Byzantine faults. In face of this limitation, traditional utility
functions must be augmented to accommodate
Byzantine-awareness. Additionally, alternative rules for predicting how the game
will be played have also been proposed to address Byzantine behaviour.

To the best of our knowledge, the work of Eliaz et. al.~\cite{eliaz}
was the first to address the issues above, 
introducing the notion of $k$-Fault Tolerant Nash Equilibrium ($k$-FTNE). In this context,
a profile of strategies is $k$-FTNE if the strategy of each player is a best response to 
the strategy of other players, independently of the identity of Byzantine players
and the arbitrary strategy they follow. This concept was later applied to
virus inoculation games\,\cite{moscibroda}. In~\cite{Wong:11}, the authors 
discuss the limitations imposed by \emph{regret freedom}
on communication games, by proving that there are no non-trivial equilibria
that provide regret-freedom strategies. Then, they propose a different approach
named \emph{regret-braving} where players are willing to obey the specified solutions
basing on their expectations about the environment, and these strategies are
regret-free as long as those expectations hold. In our work, we consider
that players are risk-averse, that is, they always hold the expectation that Byzantine
players will follow the worst possible strategy to their utility.

In practice, rational players can seek maximising their utility function by colluding 
with other players, i.e., forming coalitions. Therefore, the solution concepts are more 
robust if they account for such rational behaviour. Aumann~\cite{Aumann:59} 
addressed this issue by defining an equilibrium as a profile
of strategies where no deviating collusion strategy provides
a greater utility for all players of the group. Then, Bernheim et. al.\,\cite{Bernheim:87}
introduced the notion of \emph{coalition-proof Nash equilibrium}, where no 
deviations by a coalition can perform better, although they do not allow further deviations to the collusion strategy. 
This work was later extended to take into consideration correlated strategies\,\cite{Moreno:96}.

The work of~\cite{abraham} considered the existence of processes with
unexpected utilities and collusion. The authors proposed the solution
concept of $(k,t)$-robustness, where no process can increase its
utility by deviating in collusion with up to $k-1$ other processes,
regardless of the Byzantine behaviour of up to $t$ processes. This
notion is stronger than the previous models for collusion, since it
accounts for arbitrary collusion where it should be true that no
player performs better by deviating from the equilibrium strategy,
even if that implies decreasing the utility of other players within
the coalition. Unfortunately, in certain scenarios such as
communication games (where players incur communication costs), it was
shown that no game can be $(k,t)$-robust for
$k,t>0$\,\cite{bar-theory}.

Additional literature relevant to our results include works
on agreement in the BAR model~\cite{barb,bar-theory} and data 
dissemination~\cite{bargossip,flightpath,firespam}, which studied protocols
tolerant to the BAR model and showed in which conditions those solutions
provide Nash equilibriums. In~\cite{bar-altruism}, the authors studied the impact
of altruism on a repeated game modelled by the BAR model. All these works 
assume repeated interactions of processes 
in a cooperative service. On the other hand, our paper considers one-shot games, 
and therefore addresses the need to provide equilibrium strategies
for Rational processes to follow the specified algorithm based on incentives provided
in a single instance of \nbart.

\subsubsection{Contributions}
The first contribution of this paper consists in an algorithm that solves \nbart in 
asynchronous systems. We show that the proposed algorithm is correct,
assuming that all non-Byzantine processes follow it, for  
$\nprods \geq 2\nbyzp+1$ and $\ncons \geq \nbyzc +1$, where $\nbyzp$ 
and $\nbyzc $ are upper bounds on the number of Byzantine producers and 
Byzantine consumers respectively. We also show that the presented algorithm 
obtains asymptotically optimal bit complexity in certain scenarios. 

The second contribution consists in the game theoretic analysis of the proposed algorithm. 
Since processes incur communication costs, our algorithm cannot be $(k,t)$-robust\,\cite{bar-theory}, hence
we rely on a weaker notion of Byzantine aware utility function to account for Byzantine behaviour,
based on the notion proposed in~\cite{bar-theory}. 

Given that we cannot ensure that the players within 
a coalition follow the algorithm, we propose a new solution concept,
which is an adaptation of $k$-resilience to account for collusion 
in the following way. We define an equilibrium as a profile of strategies $\vsigma$
where members of a coalition are interested in deviating from $\vsigma$ only 
if their behaviour, as observed by other processes,
is equivalent to $\vsigma$.

We assume that the size of each group of Rational colluding processes
is bounded by a constant $\nconids = \nconidsp + \nconidsc$, where $\nconidsp$ 
is the number of members of the colluding group that are producers and
$\nconidsc$ is the number of consumers on the same group. We show that, if $\nprods \geq \max(\nbyzp,\nconidsp) + \nbyzp +1$ and 
$\ncons \geq \nbyzc + \nconidsc + 1$, then the algorithm provides such equilibrium, implying
that processes from any coalition follow a strategy that ensures that the \nbart properties
are fulfilled. An important consequence of this is that, in the absence of collusion,
the algorithm provides a Nash equilibrium.

\subsubsection{Paper Organisation}
The remainder of the paper is structured as follows. The system model and the \nbart problem are defined in Section~\ref{sec:model}. The
algorithm that solves NBART in the given model is presented in Section~\ref{sec:solution}, 
along with the proofs of correctness and a simple complexity analysis. 
In Section~\ref{sec:gametheory}, we perform the game theoretic
analysis of the algorithm.

\section{System Model}
\label{sec:model}

We assume an asynchronous system composed of $N$ \emph{processes} or
\emph{players} (we will use the term \emph{player} only when
performing the Game Theoretic analysis; in any other case, we will use
the name process). Processes are connected by a fully-connected
network and can communicate using reliable authenticated
point-to-point communication channels\,\cite{Cachin:978-3-642-15259-7}.

We make the distinction between \emph{identity},
\emph{process/player}, and \emph{coalition}.  An identity is a tuple
$(i,pk_i,sk_i)$, where $i$ is an identifier and $pk_i$ and $sk_i$ are
the corresponding public and private keys.  There is a set of
identities $\idset = \prods \cup \cons$, where $\prods$ and $\cons$
are the sets of producer and consumer identities, respectively, such
that $\#\prods=\nprods$ and $\#\cons = \ncons$.
Players are the decision-making entities of our Game Theoretic
analysis and are represented by a single identity. Therefore, when
referring to the process that holds the identity $(i,pk_i,sk_i)$, we
will simply refer to it as $i$. If $i \in \prods$, the corresponding
process is referred to as a producer, otherwise, it is called a
consumer. Finally, $\nprods+\ncons=N$.

As defined by the BAR model, a player can be Altruistic (if
it follows the algorithm), Byzantine (if its behaviour is arbitrary), or
Rational (if it follows the strategy that maximises its utility given
the expectations regarding the strategies followed by other players).
We assume that Rational processes adhere to the \emph{promptness
  principle}\,\cite{barb}, in the sense that if the expected utilities
of following the algorithm and deviating by delaying messages are
equivalent, then processes do not deviate.  It is said that a player
$i$ signs information with $sk_i$ by invoking $s_i(\mbox{\emph{data}})$.

\subsection{\nbart Problem}

The \nbart Problem can be defined as follows.  Each producer $p$
produces an arbitrarily large value $v_p$ by invoking the
deterministic function \emph{produce$(p$, $v_p)$}, such that any two
non-Byzantine producers produce the same value, named the
\emph{correct value}. Consumers must consume only one value $v$, sent
by some producer, by invoking \emph{consume$(c,v)$}.  The invocation
of this primitive proves that, indeed, $c$ consumes the value.  To
deal with Rational behaviour, we rely on the participation of an
abstract entity named Trusted Observer (TO), whose function is to
gather cryptographic information from the participants of each
transfer and reward processes according to their observable
behaviour. To assess the behaviour of each process, TO uses two
predicates \textit{hasProd}(evidence, $p$) and
\textit{hasAck}(evidence, $c$) that take as input the evidence
produced by \textit{TO} to indicate, respectively, if producer $p$
participated in \nbart and if consumer $c$ notified the reception of
the correct value.  \textit{TO} is said to eventually \emph{produce
  evidence} about the transfer if, when \emph{hasProd} and
\emph{hasAck} become true for all corresponding non-Byzantine
producers and consumers, TO eventually calls the primitive
\emph{certify}$($\textit{TO}, \textit{evidence}$)$ after that. With
these definitions, the \nbart problem is characterised by the
following properties:

\begin{itemize}
\item \textbf{\nbart 1} (\emph{Validity}): If a non-Byzantine consumer
  consumes $v$, then $v$ was produced by some non-Byzantine producer.

\item \textbf{\nbart 2} (\emph{Integrity}): No non-Byzantine consumer consumes
  more than once.

\item \textbf{\nbart 3} (\emph{Agreement}): No two non-Byzantine consumers
  consume different values.

\item \textbf{\nbart 4} (\emph{Eventual Consumption}): Eventually, every non-Byzantine
  consumer consumes a value. 
  
\item \textbf{\nbart 5} (\emph{Evidence}): TO eventually produces evidence about the transfer. 

\item \textbf{\nbart 6} (\emph{Producer Certification}):
  If producer $p$ is non-Byzantine, then \textit{hasProd(evidence, $p$)} eventually becomes \textit{true}.

\item \textbf{\nbart 7} (\emph{Consumer Certification}):
  If consumer $c$ is non-Byzantine, then \textit{hasAck(evidence, $c$)} eventually becomes \textit{true}.
\end{itemize}

\section{Asynchronous \nbart}
\label{sec:solution}

We now describe an algorithm that solves the \nbart problem in an asynchronous environment. We first
provide an overview, then proceed to the detailed description of the algorithm,
and we conclude with a theoretical analysis, where we prove the correctness of this solution
and perform a complexity analysis in terms of message and bit complexity.

\subsection{Overview of the Algorithm}
The algorithm can be briefly described as follows. Each producer $p$ owns a 
block ($b_p$) that belongs to the set of $\nprods$ blocks
obtained from the value $v$ by using Reed-Solomon codes, such that $v$ can be retrieved from any subset 
of $\nblocks$ blocks ($\nprods \geq \nblocks + \nbyzp$). Then, $p$ strives to transfer
$b_p$ along with the signature of the vector that contains the hashes of all blocks to a subset of consumers denoted by $\mconset_p$. 
Each consumer $c$ only needs to receive $\nblocks$ correct blocks and $\nbyzp+1$ signatures of the same vector
of hashes to consume the value. 
However, $c$ must continue to process any received information 
and send it to TO, which must (re-)invoke \emph{certify$($evidence$)$} 
whenever it receives new information, in order to fulfil the property \nbart-5.

\subsection{Algorithm in Depth}
The algorithm is depicted for producers in~Alg.~\ref{alg:producer}, for 
consumers in Alg.~\ref{alg:consumer1} and Alg.~\ref{alg:consumer2}, and for TO 
in Alg.~\ref{alg:to}. Producers use Reed-Solomon codes
to reduce the communication costs of transferring an arbitrarily large value.
The value $v$, whose length in bits is denoted by $l_v$, 
is split into $\nprods$ blocks  of size $\frac{l_v}{\nblocks}$,
such that any subset of $\nblocks$ blocks is sufficient to retrieve
the original value, where $1 \leq \nblocks \leq \nprods - \nbyzp$ and $\nblocks < l_v$.
There is a function $\emph{RS-ENC}(v,\nprods,\nblocks,\omega)$ that, given the correct value $v$, the number of producers $\nprods$,
the number of blocks $\nblocks$, and the word size $\omega$, returns a vector $\vec{v}$ containing
the $\nprods$ blocks, where $2^\omega > \nprods$. Let $\vec{h}_v$ denote the vector
containing the hashes of each of the blocks from $\vec{v}$. The inverse function
\emph{RS-DEC}$(\vec{v}',\nprods,\nblocks,\omega,\vec{h}_v)$ is defined as follows:
if there are at least $\nblocks$ blocks from $\vec{v}'$ whose hash is in $\vec{h}_v$, then it returns the
value $v$; otherwise, it returns $\bot$. We consider that all arithmetic 
operations are performed over elements of the Galois Field~$GF(2^\omega)$.

We consider that each process is unequivocally identified
by an index, between $0$ and $\nprods-1$ for producers, and between $0$ and $\ncons - 1$ for consumers.
Each consumer $c_j$ uses a deterministic function $\mprodset_{c_j}$ to determine the set 
of producers that are supposed to send it their blocks, defined in such a way that each consumer is related to 
exactly $\nblocks + \nbyzp$ producers (in this way distributing load among producers).
A possible mapping function is the following:
$\mprodset_{c_j} = \{ p_i \in \prods | i \in [k...(k+\nblocks + \nbyzp - 1)\ mod\ \nprods], k= j(\nblocks+\nbyzp)\ mod\ \nprods\}$.
It is useful to define the function that establishes the inverse relation $\mconset_{p_i} = \{ c_j \in \cons | p_i \in \mprodset_{c_j}\}$ for each producer $p_i$.
These definitions ensure that each consumer is able to receive at least $\nblocks$ blocks from non-Byzantine
producers, therefore being able to retrieve the correct value. In addition, the load
is distributed across the producers such that $\forall_{p \in \prods}:
\#\mconset_p = n \Rightarrow \forall_{p' \in \prods \setminus \{p\}}: n-1 \leq \#\mconset_{p'} \leq n+1$.

Each producer $p$ starts by storing the set of blocks from $\vec{v}$
by invoking \emph{RS-ENC}. Note that each producer will only be required to
transmit one of these blocks (each producer transmits a different
block). However, each producer is still required to send
$\vec{h}_v$. Therefore, each producer then sets the vector
\emph{hashes} to $\vec{h}_v$ (Alg.~\ref{alg:producer}, lines
4-7). Then, $p$ transfers its block along with $\vec{h}_v$ to all
consumers of $\mconset_p$ in a $\mblock$ message (lines 8-10), while
sending $\mhash$ messages to the remaining consumers only containing
$\vec{h}_v$ (lines 11-13). Both these messages are signed with the
public key of the producer. Notice that, in the $\mblock$ message,
it is not necessary to sign the block, for the signature of the hashes
already authenticates the block.

\begin{algorithm}[h]
\label{alg:producer}
\caption{\nbart ($p \in \prods$)}
{
\scriptsize
\begin{tabbing}
xx,\=xx\=xx\=xx\=xx\=xx\=xx\kill
01 \> \textbf{upon} init() \textbf{do} \\
02 \> \> blocks := $[\bot]^{\nprods}$;\\
03 \> \> hashes :=$[\bot]^{\nprods}$;\\
\\
04 \> \textbf{upon} \emph{produce}($p$, $v$) \textbf{do}\\
05 \> \> blocks := RS-ENC(value,$\nprods$,$\nblocks$,$\omega$);\\
06 \> \> \textbf{forall} $i \in \prods$ \textbf{do}\\
07 \> \> \> hashes[$i$] := \emph{hash}(blocks[$i$]);\\
08 \> \> signature := $s_p(\mblock||\mbox{hashes})$; \\
09 \> \> \textbf{forall} $c \in \mconset_p$ \textbf{do}\\
10 \> \> \> \emph{send}($p$, $c$, [\mblock, blocks[$p$], hashes, signature]);\\
11 \> \> signature := $s_p(\mhash||\mbox{hashes})$; \\
12 \> \> \textbf{forall} $c \in \cons \setminus \mconset_p$ \textbf{do}\\
13 \> \> \> \emph{send}($p$, $c$, [\mhash, hashes, signature]);
\end{tabbing}
}
\end{algorithm}

In turn, each consumer $c$ keeps all the received data blocks
in a vector \emph{blocks} and the received vectors of hashes (along with 
the signatures) in \emph{hashvecs}. In addition, there is a set \emph{missing}
that keeps the identities of the producers that have not yet sent any signed
information. Finally, \emph{correcthashvec} is the correct vector of hashes,
that is, the vector that is sent by at least $\nbyzp +1 $ producers,
and \emph{correctproducers} stores, for each producer, the value
$\bot$ if it has not yet sent any message, or the signature of the message
sent by the producer.

Each consumer uses the functions \emph{verifysig}($i$,$d$) and
\emph{verifyhash}($b$,$h$) to verify the signature by $i$ of $d$ 
and the hash of $b$ when compared to $h$, respectively.
Consumer $c$ is in one of three states: \emph{init},
\emph{gotHashes}, and \emph{consumed}. $c$ is in state
\emph{init} when \emph{hashvecs} does not contain
a majority ($\nbyzp+1$) of identical vectors of hashes.
The function \emph{minimumHashes} (Alg.~\ref{alg:consumer1}, lines 8-12) 
marks the transition between \emph{init} and \emph{gotHashes},
by setting \emph{correcthashvec} to a non-null value,
when the required majority of hashes is gathered by $c$.
Procedure \emph{consume-and-report} (lines 16-23) makes
the transition from \emph{gotHashes} to \emph{consumed}
when the consumer gathers at least $\nblocks$ correct blocks
and, therefore, the invocation of \emph{RS-DEC} returns a non-null
value. In this case, the consumer consumes the value (line 19)
and prepares a report intended to TO (lines 20-23), which is sent
by invoking the procedure \emph{report} (lines 13-15).
This report contains the vector \emph{correcthashvec} and the signature
of all the producers that already sent correct messages to $c$, i.e., 
messages that contained \emph{correcthashvec}. 

Whenever a consumer $c$ receives a $\mblock$ message from a producer
that belongs to $\mbox{\emph{missing}} \cap \mprodset_c$ (Alg.~\ref{alg:consumer2}, line 1),  
$c$ removes $p$ from \emph{missing} if the signature is valid (lines 2-3) and, according to
its state, performs one of the following actions: 
i) If $c$ is still in state \emph{init}, then it stores the received information
in the appropriate vectors and invokes \emph{minimumHashes} (lines 4-8),
in order to verify if it has already gathered a majority of identical vectors of hashes.
If that is the case, then $c$ invokes \emph{consume-and-report} (lines 9-10).
ii) If $c$ is in state \emph{gotHashes}, then it adds the received vector
of hashes along with the signature to \emph{hashvecs}, stores the block,
and invokes \emph{consume-and-report} (lines 11-15).
iii) If $c$ is in state \emph{consumed},
then it adds the signature of the producer to \emph{correctproducers}
and reports the information received from producers to TO (lines 16-18).

An almost identical approach is followed by $c$ whenever
it receives a $\mhash$ message, aside from the fact that in this case
$c$ does not expect to receive any block (lines 19-33). 

\begin{algorithm}[h]
\label{alg:consumer1}
\caption{\nbart ($c \in \cons$): Part I}
{
\scriptsize
\begin{tabbing}
xxx,\=xx\=xx\=xx\=xx\=xx\=xx\=xx\kill
01 \> \textbf{upon} \emph{init} \textbf{do}\\
02 \> \> value :=$\bot$;\\
03 \> \> correcthashvec := $\bot$;\\
04 \> \> hashvecs :=$[\bot]^{\nprods}$;\\
05 \> \> blocks := $[\bot]^{\nprods}$;\\
06 \> \> missing := $\prods$;\\
07 \> \> correctproducers := $[\bot]^{\nprods}$; \\
\\
08 \> \textbf{function}  \emph{minimumHashes}(hashvecs) \textbf{is} \\
09 \> \> \textbf{if} $\exists h: \#\{p | \mbox{hashvecs}[p] = \langle h,* \rangle\} \geq  \nbyzp + 1$ \textbf{then} \\
10 \> \> \> \textbf{return} h; \\
11 \> \> \textbf{else} \\
12 \> \> \> \textbf{return} $\bot$; \\
\\
13 \> \textbf{procedure}  \emph{report} \textbf{is} \\
14 \> \> signature := $s_c$(\mreport$||$correcthashvec$||$correctproducers);\\
15 \> \> \emph{send}($c$, \textit{TO}, [\mreport, correcthashvec, correctproducers, signature]);\\
\\
16 \> \textbf{procedure}  \emph{consume-and-report} \textbf{is} \\
17 \> \> value  := RS-DEC($\mbox{blocks},\nprods,\nblocks,\omega,\mbox{correcthashvec}$);\\
18 \> \> \textbf{if} value $\neq \bot$ \textbf{then} \\
19 \> \> \> \emph{consume}($c$, value);\\
20 \> \> \> \textbf{forall} $p \in \prods $ \textbf{do}\\
21 \> \> \> \> \textbf{if} hashvecs[p] = $\langle$correcthashvec,signature$\rangle$ \textbf{then} \\
22 \> \> \> \> \> correctproducers[p] := signature;\\
23 \> \> \> \emph{report} ();
\end{tabbing}
}
\end{algorithm}

\begin{algorithm}[h]
\label{alg:consumer2}
\caption{\nbart ($c \in \cons$): Part II}
{
\scriptsize
\begin{tabbing}
xxx,\=xx\=xx\=xx\=xx\=xx\=xx\=xx\kill
01 \> \textbf{upon}  \emph{deliver}($p$, $c$, [\mblock, pblock, phashes, msgsig])~$\land$~$p \in \mbox{missing} \cap \mprodset_c$ \textbf{do}\\
02 \> \> \textbf{if} \emph{verifysig}($p$, \mblock$||$phashes, msgsig) \textbf{then} \\
03 \> \> \> missing := missing~$\setminus$~$\{p\}$;\\
04 \> \> \> \textbf{if} \emph{verifyhash}(pblock, phashes[$p$]) \textbf{then}\\
05 \> \> \> \> \textbf{if} correcthashvec $= \bot$ \textbf{then}\\
06 \> \> \> \> \> hashvecs[p] := $\langle$phashes, msgsig$\rangle$; \\
07 \> \> \> \> \> blocks[p] := pblock; \\
08 \> \> \> \> \> correcthashvec := \emph{minimumHashes}(hashvecs); \\
09 \> \> \> \> \> \textbf{if} correcthashvec $\neq \bot$ \textbf{then} \\
10 \> \> \> \> \> \> \emph{consume-and-report} ();\\
11 \> \> \> \> \textbf{else if} value $= \bot$ \textbf{then} \\
12 \> \> \> \> \> \textbf{if} phashes = correcthashvec \textbf{then} \\
13 \> \> \> \> \> \> hashvecs[p] := $\langle$phashes, msgsig$\rangle$; \\
14 \> \> \> \> \> \> blocks[p] := pblock; \\
15 \> \> \> \> \> \> \emph{consume-and-report} ();\\
16 \> \> \> \> \textbf{else if} phashes = correcthashvec \textbf{then} \\
17 \> \> \> \> \> correctproducers[$p$] := msgsig;\\
18 \> \> \> \> \> \emph{report} ();\\
\\
19 \> \textbf{upon}  \emph{deliver}($p$, $c$, [\mhash, phashes, msgsig])~$\land$~$p \in \mbox{missing} \cap \prods \setminus \mprodset_c$  \textbf{do}\\
20 \> \> \textbf{if} \emph{verifysig}($p$, \mhash$||$phashes, msgsig) \textbf{then} \\
21 \> \> \> missing := missing~$\setminus$~$\{p\}$;\\
22 \> \> \> \textbf{if} correcthashvec $= \bot$ \textbf{then} \\
23 \> \> \> \> hashvecs[p] := $\langle$phashes, msgsig$\rangle$; \\
24 \> \> \> \> correcthashvec := \emph{minimumHashes}(hashvecs); \\
25 \> \> \> \> \textbf{if} correcthashvec $\neq \bot$ \textbf{then} \\
26 \> \> \> \> \> \emph{consume-and-report} ();\\
27 \> \> \> \textbf{else if} value $= \bot$  \textbf{then} \\
28 \> \> \> \> \textbf{if} phashes = correcthashvec \textbf{then} \\
29 \> \> \> \> \> hashvecs[p] := $\langle$phashes, msgsig$\rangle$; \\
30 \> \> \> \> \> \emph{consume-and-report} ();\\
31 \> \> \> \textbf{else if} phashes = correcthashvec \textbf{then} \\
32 \> \> \> \> correctproducers[$p$] := msgsig;\\
33 \> \> \> \> \emph{report} ();
\end{tabbing}
}
\end{algorithm}

The trusted observer only waits for $\mreport$ messages from
consumers to include all the received information in the array \emph{evidence} (lines 3-5).
In addition, TO repeatedly tries to produce the evidence about the
transfer whenever it receives new information (line 6).

\begin{algorithm}[h]
\label{alg:to}
\caption{\nbart (trusted observer \textit{TO})}
{
\scriptsize
\begin{tabbing}
xx,\=xx\=xx\=xx\=xx\=xx\=xx\kill
01 \> \textbf{upon} \emph{init} \textbf{do}\\
02 \> \> evidence := $[\bot]^{\ncons}$;\\
\\
03 \> \textbf{upon}  \emph{deliver}($c$, \textit{TO}, [\mreport, hashesvec, producers, signature]) \textbf{do}\\
04 \> \> \textbf{if} \emph{verifySig}($c$, \mreport$||$hashesvec$||$producers, signature) \textbf{then}\\
05 \> \> \> evidence[$c$] := $\langle$hashesvec,producers$\rangle$;\\
06 \> \> \> \emph{certify}(\textit{TO}, evidence);
\end{tabbing}
}
\end{algorithm}

\subsection{Predicates}
We now define the predicates \emph{hasProd} and \emph{hasAck}.
It is said that producer $p$ \emph{is certified} by consumer $c \in \mconset_p$ iff
$\mbox{evidence}[c] = \langle \vec{h}_v, \mbox{report}\rangle$ and $\mbox{report}[p] = s_p(\mblock,\vec{h}_v)$.
We say that producer $p$ is certified by consumer $c \in \cons \setminus \mconset_p$ iff
$\mbox{evidence}[c] = \langle \vec{h}_v, \mbox{report}\rangle$ and $\mbox{report}[p] = s_p(\mhash,\vec{h}_v)$.
Let $\bar{\prods} \subseteq \prods$ and 
$\bar{\cons} \subseteq \cons$ be the greatest sets that fulfil the following conditions:
i) for each $p \in \bar{\prods}$ and $c \in \bar{\cons}$, $p$ is certified by $c$;
and ii) for each $c \in \bar{\cons}$, $c$ invokes \emph{consume}($c$,$v$).

With this in mind, we now define the predicates as follows:
\begin{itemize} 
 \item For the predicates to be true for any process, $\#\bar{\prods} \geq \nprods - \nbyzp$ and $\#\bar{\cons} \geq \ncons - \nbyzc$; 
 \item \emph{hasProd}(evidence,$p$) is true iff $p \in \bar{\prods}$;
 \item \emph{hasAck}(evidence,$c$) is true iff $c \in \bar{\cons}$.
\end{itemize}

\subsection{Correctness}
In this section, the correctness of the above algorithm is proven in an asynchronous
environment, assuming that $\nprods \geq 2\nbyzp+1$, $\ncons \geq \nbyzc+1$,
and that all non-Byzantine processes follow the algorithm.
In the following two lemmas, we start by showing that the consumers eventually gather enough 
information to consume the correct value.

\begin{lemma}
\label{lemma:minimumHashes}
For each non-Byzantine consumer $c \in \cons$, 
\emph{minimumHashes} eventually returns exactly one vector $\vec{h}^* \neq \bot$
and $\vec{h}^* = \vec{h}_v$.
\end{lemma}

\begin{proof}
Every non-Byzantine consumer receives $\vec{h}_v$ from all non-Byzantine producers, eventually.
Since $\nprods \geq 2\nbyzp+1$, only the vector $\vec{h}_v$ can
be sent by $\nbyzp + 1$ producers. 
Therefore, \emph{minimumHashes} only returns a non-null vector $\vec{h}^*$ 
if $\vec{h}^* =\vec{h}_v$ and this occurs eventually. Also, when \emph{correcthashvec} becomes
non-null, $c$ never invokes \emph{minimumHashes} again.
\end{proof}

\begin{lemma}
\label{lemma:consume}
For each consumer $c \in \cons$, $c$ eventually invokes \emph{consume}($c$,$v$),
and only once.
\end{lemma}

\begin{proof}
It follows from Lemma~\ref{lemma:minimumHashes} that, for each non-Byzantine consumer $c$, \emph{correcthashvec}
is eventually set to $\vec{h}_v$, and $c$ eventually starts invoking \emph{consume-and-report}. 
By the fact that producers send their blocks to all consumers of $\mconset$, and by the conditions
$\nprods \geq \nblocks + \nbyzp$ and $\#\mprodset_c = \nblocks + \nbyzp$,
$c$ eventually receives $\nblocks$ blocks, correct according to $\vec{h}_v$. Hence, RS-ENC eventually returns $v$,
and only $v$ by the property of non-collision of hash functions. A trivial inspection of the algorithm shows
that, once \emph{value} is set to $v \neq \bot$, $c$ consumes $v$ and never invokes \emph{consume-and-report}
again.
\end{proof}

\begin{lemma}
\label{lemma:pcertaux}
For each non-Byzantine producer $p$ and each non-Byzantine consumer $c \in \mconset_p$, 
eventually $c$ certifies $p$.
\end{lemma}

\begin{proof}
According to the algorithm, $p$ always sends a $\mblock$ message containing its block and $s_p(\mblock||\vec{h}_v)$
to all $c \in \mconset_p$, whereas $p$ sends a $\mhash$ message to
all  $c \in \cons \setminus \mconset_p$, containing $s_p(\mhash||\vec{h}_v)$.
If $c$ receives this information when it is still in one of the states \emph{init}
and \emph{gotHashes}, then, by Lemma~\ref{lemma:consume}, 
$c$ eventually sends a report to TO containing this information. If $c$ is already in state \emph{consumed},
then $c$ immediately sends the report containing this information
when it receives the message from $p$. Either way, $c$ eventually certifies $p$.
\end{proof}

\begin{lemma}
\label{lemma:certsets}
There exist sets $\bar{\prods}$ and $\bar{\cons}$ of non-Byzantine producers and non-Byzantine consumers,
respectively, such that: i) $\#\bar{\prods} \geq \nprods - \nbyzp$ and $\#\bar{\cons} \geq \ncons - \nbyzc$;
ii) for each $p \in \bar{\prods}$ and $c \in \bar{\cons}$, $c$ eventually certifies $p$;
and iii) for each $c \in \bar{\cons}$, $c$ eventually invokes \emph{consume}($c$,$v$), where
$v$ is the correct value.
\end{lemma}

\begin{proof}
i) follows from the fact that there are $\nprods - \nbyzp$ non-Byzantine producers
and $\ncons - \nbyzc$ non-Byzantine consumers; ii) follows from Lemma~\ref{lemma:pcertaux};
and iii) follows from Lemma~\ref{lemma:consume}.
\end{proof}

The next theorem concludes the proofs of correctness by showing that
each \nbart property is fulfilled by the presented algorithm.

\begin{theorem}
\label{theorem:correctness}
The proposed algorithm solves \nbart in an asynchronous environment,
assuming that all non-Byzantine processes follow the algorithm.
\end{theorem}

\begin{proof}
The proof is performed individually for each property:
\begin{itemize}
\item (\emph{Validity}): By Lemmas~\ref{lemma:minimumHashes} and~\ref{lemma:consume} and by the non-collision property
of hash functions, $c$ consumes the correct value, which is produced by all non-Byzantine producers.

\item (\emph{Integrity}): Follows from Lemma~\ref{lemma:consume}.

\item (\emph{Agreement}): It follows directly from \emph{Validity} and the fact
that all non-Byzantine producers send a block corresponding
to the same value.

\item (\emph{Eventual Consumption}): Follows from Lemma~\ref{lemma:consume}.

\item (\emph{Evidence}): TO invokes \emph{certify}(evidence) whenever it receives
new information, either from producers or consumers. Thus,
whenever \emph{hasProd}(evidence,$p$)
and \emph{hasAck}(evidence,$c$) become true for each non-Byzantine
producer $p$ and non-Byzantine consumer $c$ respectively, TO invokes
\emph{certify}(evidence).

\item (\emph{Producer} and \emph{Consumer Certification}): Follows from 
Lemma~\ref{lemma:certsets}.

\end{itemize}
\end{proof}

\subsection{Complexity Analysis}
The algorithm is evaluated in terms of message and bit complexity.
The message complexity is $O(\nprods \ncons)$ due to $\ncons$ messages
sent by each producer that contain the signature of $\vec{h}_v$. However, since
the value may be arbitrarily large, the size of each message may vary significantly,
so it is interesting to also evaluate the number of bits exchanged,
that is, the bit complexity. For this analysis,
let $l_v$, $l_s$, and $l_h$ denote the bit length of the value, a signature
and an hash. The bit complexity is $O(\ncons (\nblocks + \nbyzp) \frac{l_v}{\nblocks} + \nprods\ncons( l_s + \nprods l_h))$.
Notice that, if $\nblocks \geq O(\nbyzp)$ and $l_v \gg l_s,l_h$, then the bit complexity is $O(\ncons)$, 
which is asymptotically optimal, since there must be at least a value transfer per consumer.

\section{Game Theoretic Analysis}
\label{sec:gametheory}

The purpose of this analysis is to show that it is in every Rational process interest
to follow the algorithm. We take into consideration some degree of arbitrary collusion.

\subsection{Definitions}
The algorithm is modelled as a coalitional game $\Gamma=(\idset,\collset,\Sigma_\idset,(\succeq_t)_{t \in \collset},(u_i)_{i \in \idset})$:

\begin{itemize}
 \item $\idset = \prods \cup \cons \cup \{\mbox{TO}\}$ is the set of players.
 
 \item $\collset$ is the set of non-empty subsets of $\idset \setminus \{\mbox{TO}\}$, which contains all the possible coalitions. 
 Each coalition $t \in \collset$ may contain simultaneously producers and consumers,
represented by $t_\prods= t \cap \prods$ and $t_\cons = t \cap \cons$, respectively.

 \item $\Sigma_\idset$ is a set containing all the profile of pure strategies $\vsigma_\idset$ followed by all players
of $\idset$. $\Sigma_t$ for $t \in \collset$ denotes the set of all collusion strategies the players of $t$ may follow.
 \item $\succeq_t$ is a preference relation on $\Sigma_\idset \times \Sigma_\idset$. We assume that $\succeq_t$ is transitive and reflexive.
We can define the relation of strict preference $\succ_t$ as:  for any two profiles of strategies $\vsigma_\idset^*,\vsigma_\idset' \in \Sigma_\idset$, 
$\vsigma_\idset^* \succ_t \vsigma_\idset'$ iff $\neg(\vsigma_\idset' \succeq_t \vsigma_\idset^*)$.
If $\vsigma_\idset^* \succ_t \vsigma_\idset'$, then all the players of $t$ will always follow $\vsigma_\idset^*$ over $\vsigma_\idset' $.

 \item $u_i$ is the utility function of each player $i \in \idset$, 
defined as $u_i(\vsigma_\idset) = \beta_i(\vsigma_\idset) - \alpha_i(\vsigma_\idset)$, 
where $\beta_i(\vsigma_\idset)$ are the benefits and $\alpha_i(\vsigma_\idset)$ the costs
$i$ incurs when players obey $\vsigma_\idset$. 
\end{itemize}

Sometimes, we will denote the composition
of two profiles $\vsigma_A$ and $\vsigma_B$ as $\vsigma_{A \cup B} = (\vsigma_A,\vsigma_B)$,
where $A$ and $B$ are any two disjoint sets of players. Conversely, $u_i(\vsigma_A,\vsigma_B)$ is
equivalent to $u_i(\vsigma_{A \cup B})$. Each producer $p$ obtains a benefit $\pbenef$ 
iff \emph{hasProd}(evidence,$p$) eventually becomes true, whereas each consumer $c$ 
obtains a benefit $\cbenef$ iff \emph{hasAck}(evidence,$c$) eventually becomes true. 
It is assumed that for all $p \in \prods$, $\pbenef > \alpha_p(\vsigma_\idset)$, 
and for all $c \in \cons$, $\cbenef > \alpha_c(\vsigma_\idset)$, where $\vsigma_\idset$ is 
the profile of strategies where all players follow the algorithm.

A coalition $t$ is said to be Rational if the preference relation $\succeq_t$ fulfils the following condition:

$$ \forall_{i \in t} \forall_{\vsigma_\idset \in \Sigma_\idset, \vsigma^*_t \in \Sigma_t} u_i(\vsigma_{\idset}) \geq u_i(\vsigma_t^*, \vsigma_{\idset \setminus t}) 
\Rightarrow (\vsigma_t, \vsigma_{\idset \setminus t}) \succeq_t (\vsigma_t^*, \vsigma_{\idset \setminus t}).$$

We assume that the same relation holds, by only replacing $\geq$ for $>$ and $\succeq_t$ for $\succ_t$.
It follows that if $\#t = 1$ and the only player $i \in t$ is Rational, then
for any two profiles of strategies $\vsigma_\idset^*,\vsigma_\idset' \in \Sigma_\idset$, 
$\vsigma_\idset^* \succeq_t \vsigma_\idset'$ iff $u_i(\vsigma_\idset^*) \geq u_i(\vsigma_\idset')$.
On the contrary, if $\#t=1$ and the player $i \in t$ is Altruistic, then $t$ is also said to
be Altruistic and it is true that $(\vsigma_t,\vsigma_{\idset \setminus t}^*) \succ_t (\vsigma_\idset^*)$ 
for all $\vsigma_\idset^* \in \Sigma_\idset$ and considering that $\vsigma_\idset$ denotes the profile
of strategies where all players follow the algorithm. In any other case, $t$ is Byzantine, implying that $\succeq_t$ is arbitrary
due to the Byzantine behaviour of some player from $t$. It is important to notice that, if $t$ is Byzantine,
then all players of $t$ are also considered to be Byzantine, even if some of them have Rational intentions.
A coalition $t$ is said to be a producer ($t \in \collset_\prods$) if $t_\prods \neq \emptyset$
and it is said to be a consumer ($t \in \collset_\cons$) if $t_\cons \neq \emptyset$. 
The purpose of these definitions is to model scenarios of arbitrary collusion where, for instance, 
a producer $p$ never executes any local function to produce the value. Instead, it
requests the hash of the blocks to other player $i$, signs this information, and sends it
to $i$. Then, $i$ may transfer the block and the signature of $p$ to all the consumers
that expect this information, as if it were sent by $p$.

For simplicity, we model Byzantine behaviour as a single coalition composed by 
up to $\nbyzp + \nbyzc$ players.
We consider an arbitrary number of non-Byzantine coalitions, 
as long as each coalition is never composed
by more than $\nconidsp$ producers and $\nconidsc$ consumers.
The distinction between producers and consumers
will allow us a more refined analysis of the bounds on the minimum
number of producers and consumers. If we only considered a single parameter, the bounds would be stricter than necessary. As it will be shown later, we now require the following 
conditions to hold for the algorithm to be
tolerant to collusion: $\nprods \geq \max(\nbyzp,\nconidsp) + \nbyzp +1$ 
and $\ncons \geq \nbyzc + \nconidsc +1$.

\subsection{Expected Utility and Solution Concept}

We use the notion of Byzantine-aware utility function for risk-averse players
introduced in~\cite{bar-theory}. An improvement of this work for models
where players may be risk-seekers is left for future work. Let $\byzp$ and $\byzc$ denote the set of Byzantine
producers and consumers, respectively, and let $\vec{\pi}_\prods \in \Pi_\prods$ and
$\vec{\pi}_\cons \in \Pi_\cons$ be the corresponding profiles of strategies. 
Let us denote by $\vsigma_{\idset \setminus \byz, \vec{\pi}_\cons,\vec{\pi}_\cons}$
the profile of strategies where all non-Byzantine players follow the strategy specified by $\vsigma_\idset$,
Byzantine producers follow the strategies of $\vec{\pi}_\prods$ and Byzantine consumers
obey the strategies of $\vec{\pi}_\cons$. The expected
utility of each player $i \in \idset \setminus \byz$ is defined as follows:

\begin{equation}
\label{eq:util}
\exputil_i(\vsigma_\playerset) = \min_{\byzp:\#\byzp \leq \nbyzp,\byzc: \#\byzc \leq \nbyzc}  \circ
 \min_{\vec{\pi}_\prods \in \Pi_\prods,\vec{\pi}_\cons \in \Pi_\cons} \circ  u_i(\vsigma'_{M \setminus \byz,\vec{\pi}_\cons,\vec{\pi}_\cons}).
\end{equation}

Recall that, since we consider communication costs, a solution concept as strong as $(k,t)$-robustness is impossible in our case.
To overcome this impossibility result, we use the concept of $k$-resilience combined
with the Byzantine aware utility function defined above.
However, we still cannot ensure that no player from a coalition $t$
can increase its utility regardless of whether some other player 
obtains a lower utility or not. What we intend to show is that, regardless of
the preferred collusion strategy of each coalition, the chosen strategies fulfil
the \nbart properties.

In order to formalise this intuition, we define the \emph{observable behaviour}
of each coalition $t \in \collset$ for the profile of strategies $\vsigma_t$ as a multi-set
of events triggered in each player $i \in \idset \setminus t$ that are influenced by $\vsigma_t$,
which we denote by $\phi_i(\vsigma_t)$. For any player $i \in \idset \setminus t$,
the delivery of a message sent by some player $j \in t$ is an event.
In addition, there are two events triggered in TO, namely \emph{produce}($p$,$v$)
for each $p \in t_\prods$ and \emph{consume}($c$,$v$) for each $c \in t_\cons$.
Henceforth, the meaning of a producer producing a value or a consumer
consuming a value is that the corresponding event is eventually triggered in TO.

We say that collusion profile $\vsigma_t^* \in \Sigma_t$ is 
compliant with the profile $\vsigma _\idset=(\vsigma_t,\vsigma_{\idset \setminus t})$
if $\forall_{i \in \idset \setminus t} \phi_i(\vsigma_t^*)=\phi_i(\vsigma_t)$.
The set of profiles of strategies compliant with $\vsigma_\idset$ is denoted by $\Sigma_t(\vsigma_\idset)$,
where $\vsigma_t \in \Sigma_t(\vsigma_\idset)$.
The solution concept we use in this work, named \emph{$n$ 
collusion tolerance} ($n$-cotolerance), 
is similar to the concept of $k$-resilience, aside from
the fact that we do not require that players in collusion follow the algorithm
exactly; only that they follow a profile of strategies from $\Sigma_t(\vsigma_\idset)$. 
More precisely:

\begin{definition}
\label{def:cotolerance}
For any $n \in \mathbb{N}$, a profile of strategies $\vsigma_\idset$ is $n$-cotolerant
iff for all $t \in \collset$ such that $\#t \leq n$, for all $\vsigma_t^* \in \Sigma_t(\vsigma_\idset)$ such that 
$(\vsigma_t^*,\vsigma_{\idset \setminus t}) \succeq_t \vsigma_\idset$, 
and for all $\vsigma_t' \in \Sigma_t \setminus \Sigma_t(\vsigma_\idset)$,
$(\vsigma_t^*,\vsigma_{\idset \setminus t}) \succ_t (\vsigma_t',\vsigma_{\idset \setminus t})$.
\end{definition}

The above definition is generic and may be of independent interest. In order to apply it to the \nbart problem, 
we additionally need to capture the distinction between producers and consumers. Therefore,
we introduce two parameters $x,y \in \mathbb{N}$, that establish the limit on the number of 
producers and consumers within the coalition respectively, such that $n \geq x, y$ 
and $n \leq x+y$. With this definition, if $n =1$, then there is no collusion among non-Byzantine players. 
Henceforth, we will say that a profile of strategies $\vsigma_\idset$ is $(n,x,y)$-cotolerant iff it is $n$-cotolerant, $n \geq x, y$ 
and $n \leq x+y$, and for all $t \in \collset$ $\#t_\prods \leq x$ and $\#t_\cons \leq y$.

\subsection{Tolerance to Collusion}
The purpose of this section is twofold: i) show that, considering that $\vsigma_\idset$ denotes
the profile of strategies where all players follow the algorithm, for any combination
of Byzantine and Rational collusions, and any coalition $t$, if all players of $t$ follow
a profile of strategies from $\Sigma_t(\vsigma_\idset)$, then the \nbart properties are fulfilled;
and ii) show that any profile of strategies $\vsigma_t^* \in \Sigma_t$ is preferable to $\vsigma_t$
only if $\vsigma_t^* \in \Sigma_t(\vsigma_\idset)$. The proofs of this section rely
on the assumption that $\nprods \geq \max(\nbyzp,\nconidsp) + \nbyzp + 1$
and $\ncons \geq \nbyzc + \nconidsc + 1$.

We find it useful to identify the following corollary that states that in any coalition $t$,
\emph{produce}($p$,$v$) must be invoked for all $p \in t_\prods$, which follows from 
the fact that $\mbox{\emph{produce}}(p,v) \in \phi_{\mbox{TO}}(\vsigma_t)$.

\begin{corollary}
\label{corollary:mustproduce}
For any $t \in \collset_\prods$, if $t$ follows a profile of strategies from $\Sigma_t(\vsigma_\idset)$,
then for each $p \in t_\prods$, $p$ invokes \emph{produce}$(p,v)$.
\end{corollary}

Let $\epsilon : \Sigma_t \rightarrow {\cal E}^{\ncons}$
be a function that for each profile $\vsigma_\idset^*$ returns an instance of the data structure
$\mbox{evidence} \in {\cal E}^{\ncons}$ stored by TO when it produces evidence about the transfer,
by replacing any entrance corresponding to a Byzantine player by the value $\bot$, i.e., 
if $c \in \byzc$ and $e = \epsilon(\vsigma_\idset^*)$, then $e[c] = \bot$, and if $p \in \byzp$,
then $e[c][p]=\bot$ for all $c \in \cons \setminus \byzc$.

We state the following proposition that $\epsilon$ depends only on the observable behaviour of
each player:
\begin{proposition}
\label{prop:evidence}
For any $t \in \collset$ and for any profile $\vsigma_\idset^* \in \Sigma_\idset$,
if $\phi_{\mbox{TO}}(\vsigma_\idset^*) = \phi_{\mbox{TO}}(\vsigma_{\idset})$,
then $\epsilon(\vsigma_\idset^*)=\epsilon(\vsigma_\idset)$.
\end{proposition}

We show in the following theorem that if each coalition $t$ follows a strategy from $\Sigma_t(\vsigma_t^*)$,
then the algorithm tolerates collusion. Fix any arbitrary $f \in \byz$ such that $\#f_\prods \leq \nbyzp$ and $\#f_\cons \leq \nbyzc$.
By assumption, $\nprods \geq \max(\nbyzp,\nconidsp) + \nbyzp +1$
and $\ncons \geq \nbyzc + \nconidsc +1$, and let $l$ denote an arbitrary partition 
of $\idset \setminus (\{\mbox{TO}\} \cup f)$ such that, for any $t \in l$,
$\#t_\prods \leq \nconidsp$ and $\#t_\cons \leq \nconidsc$. We use the notation $\#(e,m)$ to denote
the frequency of element $e$ in the multi-set $m$.

\begin{theorem}
\label{theorem:goodalternative}
For some arbitrary partition $l$, and for any $\vsigma_\idset^*= ((\vsigma_t^*)_{t \in l,\vsigma_t^* \in \Sigma_t(\vsigma_\idset)},(\pi_p)_{p \in f_\prods},(\pi_c)_{c \in f_\cons})$,
if all players follow $\vsigma_\idset^*$, then the \nbart properties are fulfilled.
\end{theorem}

\begin{proof}
Notice that, in this scenario, it is also true that: 1) $\nprods \geq 2\nbyzp + 1$ and 2) $\ncons \geq \nbyzc+1$. Let us fix some arbitrary $t \in \collset_\cons \cap l$ and $c \in t_\cons$.
The correctness is proved for each of the \nbart properties:
\begin{itemize}

 \item (\emph{Validity}): \emph{consume}($c$,$v$)~$\in$~$\phi_{\mbox{TO}}(\vsigma_t^*)$,
where $v$ must be a value for which there are $\nbyzp+1$ signatures of $\vec{h}_v$, otherwise $\phi_{\mbox{TO}}(\vsigma_\idset^*) \neq \phi_{\mbox{TO}}(\vsigma_\idset)$ 
and $\vsigma_t^* \notin \Sigma_t(\vsigma_\idset)$. By 1) and Corollary~\ref{corollary:mustproduce}, there is
only one value that fulfils these restrictions, which is the value produced by all non-Byzantine producers.

 \item (\emph{Integrity}): 
 Since the players of $t$ follow $\vsigma_t^*$ and $\vsigma_t^* \in \Sigma_t(\vsigma_\idset)$,  $\#(\mbox{\emph{consume}}(c,v),\phi_{\mbox{TO}}(\vsigma_t^*)) = 1$.
 
 \item (\emph{Agreement}):  It follows directly from \emph{Validity} and Corollary~\ref{corollary:mustproduce}.
 
 \item (\emph{Eventual Consumption}): Since $\phi_c(\vsigma_\idset^*)=\phi_c(\vsigma_\idset)$, $t$ receives
blocks from all the non-Byzantine producers from $\mprodset_c \setminus t_\prods$. 
If $\#(\mprodset_c \setminus t_\prods) \geq \nbyzp + \nblocks$,
then $c$ eventually gathers $\nblocks$ blocks corresponding to the correct value,
otherwise, $\#t_\prods \geq 1$ and by Corollary~\ref{corollary:mustproduce} some producer of 
$t$ produces the value. In either case, by the definition of $\Sigma_t(\vsigma_\idset)$, 
$c$ must invoke \emph{consume}($c$,$v$).

 \item (\emph{Evidence}): It follows from the fact that TO is Altruistic.
 
 \item (\emph{Producer and Consumer Certification}): By the definition
 of $\Sigma_t(\vsigma_t^*)$, $\phi_{\mbox{TO}}(\vsigma_\idset^*) = \phi_{\mbox{TO}}(\vsigma_\idset)$.
By Theorem~\ref{theorem:correctness} and by 1) and 2), if all players follow $\vsigma_\idset$,
then the properties \nbart 6-7 are fulfilled for $e=\epsilon(\vsigma_\idset)$. It follows from Proposition~\ref{prop:evidence} 
that $\epsilon(\vsigma_\idset^*) = e$. Since the value of the predicates only depends on $e$,
then these properties also hold in this new scenario.

\end{itemize}
\end{proof}

We now provide the proofs that the profile of strategies $\vsigma_\idset$
where all players follow the algorithm is $(\nconidsp + \nconidsc,\nconidsp,\nconidsc)$-cotolerant
for $\nprods \geq \max(\nbyzp,\nconidsp) + \nbyzp +1$ and
$\ncons \geq \nbyzc + \nconidsc + 1$. The following two lemmas show
that, for each $t \in \collset$ the expected benefit is 0 for all $i \in t$, whenever
players of $t$ follow a profile of strategies from $\Sigma_t \setminus \Sigma_t(\vsigma_\idset)$. Recall that we assume that the players are risk averse. Therefore, the analysis is done assuming worst case Byzantine behaviour.
 
\begin{lemma}
\label{lemma:certallplayers}
For any $t \in \collset_\cons$,
let $\vsigma_t' \in \Sigma_t \setminus \Sigma_t(\vsigma_\idset)$
be any profile of strategies where $t$ does not ensure that for all
$c \in t_\cons$ and $p \in \prods$, $c$ certifies $p$ and invokes \emph{consume}$(c,v)$, 
and, for each, $p \in t_\prods$ $p$ invokes \emph{produce}$(p,v)$. 
Then, for all $i \in t$, $\expbenef_i(\vsigma_t',\vsigma_{\idset \setminus t}) = 0$.
\end{lemma}

\begin{proof}
Assume worst case Byzantine behaviour. If $n \geq 1$ producers are not certified by all consumers of $t$,
those $n$ producers are certified by less than $\ncons - \nbyzc$ consumers.
Since $\#t_\prods < \nprods - \nbyzp$, $\#\bar{\prods} \leq \nprods - n- \nbyzp < \nprods - \nbyzp$.
Conversely, if consumers from $t_\cons$ do not consume the correct
value, then it is true that $\#\bar{\cons} < \ncons - \nbyzc$,
due to the fact that $\#t_\cons < \ncons - \nbyzc$.
By the definition of the predicates, for all $p \in t_\prods$ and $c \in t_\cons$,
\emph{hasProd}(evidence,$p$) and \emph{hasAck}(evidence,$c$)
are false. Therefore, for all $i \in t$ $\expbenef_i(\vsigma_t',\vsigma_{\idset \setminus t}) = 0$.
\end{proof}

\begin{lemma}
\label{lemma:badalternative}
For any $t \in \collset$, let $\vsigma_t' \in \Sigma_t \setminus \Sigma_t(\vsigma_\idset)$.
Then, for all $i \in t$, $\expbenef_i(\vsigma_t',\vsigma_{\idset \setminus t}) = 0$.
\end{lemma}

\begin{proof}
Assume worst case Byzantine behaviour. By the definition of $\Sigma_t(\vsigma_\idset)$, there exists $j \in \idset \setminus (\byz \cup t)$
such that $\phi_j(\vsigma_t') \neq \phi_j(\vsigma_t)$, which implies that not all expected events are triggered in $j$ for some player $i \in t$, some consumer 
does not consume the correct value, or some producer does not produce the correct value. If $j$ is a consumer or a producer, then it follows directly from Lemma~\ref{lemma:certallplayers}
that, for all $i \in t$, $\expbenef_i(\vsigma_t',\vsigma_{\idset \setminus t}) = 0$. If $j$ is
TO, then either 1) $i$ is a producer, and $i$ is not certified by some consumer or does not produce the value; or
2) $i$ is a consumer, and $i$ does not certify some player or does not consume the value. 
In both cases, by Lemma~\ref{lemma:certallplayers},
it is true that for all $i \in t$, $\expbenef_i(\vsigma_t',\vsigma_{\idset \setminus t}) = 0$.
\end{proof}

The following theorem concludes that the proposed algorithm is $(\nconidsp+\nconidsc,\nconidsp,\nconidsc)$-cotolerant.

\begin{theorem}
\label{theorem:cotolerant}
Let $\vsigma_\idset \in \Sigma_\idset$ denote the profile of strategies
where all players follow the algorithm. Then, $\vsigma_\idset$
is $(\nconidsp+\nconidsc,\nconidsp,\nconidsc)$-cotolerant.
\end{theorem}

\begin{proof}
Let $t \in \collset$ be any coalition such that $\#t_\prods \leq \nconidsp$
and $\#t_\cons \leq \nconidsc$. By Theorem~\ref{theorem:goodalternative},
for all $p \in t_\prods$, $\expbenef_p(\vsigma_\idset) = \pbenef$
and for all $c \in t_\cons$, $\expbenef_c(\vsigma_\idset) = \cbenef$.
Therefore, for all $i \in t$, $\expbenef_i(\vsigma_\idset) > \expcost_i(\vsigma_\idset)$ and $\exputil_i(\vsigma_{\idset}) > 0$.
Furthermore, it follows from Lemma~\ref{lemma:badalternative} 
that for all $\vsigma_t' \in \Sigma_t \setminus \Sigma_t(\vsigma_\idset)$,  $\expbenef_i(\vsigma_t',\vsigma_{\idset \setminus t}) = 0$.
Therefore, $\exputil_i(\vsigma_t',\vsigma_{\idset \setminus t}) \leq 0 < \exputil_i(\vsigma_\idset)$,
which implies that $\vsigma_\idset \succ_t (\vsigma_t',\vsigma_{\idset \setminus t})$.
Consequently, for all $\vsigma_t^* \in \Sigma_t(\vsigma_\idset)$, if $(\vsigma_t^*,\vsigma_{\idset \setminus t}) \succeq_t \vsigma_\idset$,
then $(\vsigma_t^*,\vsigma_{\idset \setminus t}) \succ_t (\vsigma_t',\vsigma_{\idset \setminus t})$.
This allows us to conclude that $\vsigma_\idset$ is $(\nconidsp+\nconidsc,\nconidsp,\nconidsc)$-cotolerant.
\end{proof}

\subsection{Discussion}
\label{sec:discussion}
Some important consequences result from Theorems~\ref{theorem:goodalternative} 
and~\ref{theorem:cotolerant}. One is that $\vsigma_\idset$ is $(1,1,1)$-cotolerant. By the definition of $\succeq_t$
for any $t \in \collset$ such that $\#t = 1$ and by the fact that $\Sigma_t(\vsigma_\idset)=\{\vsigma_\idset\}$, 
$\vsigma_\idset$ is a \emph{Nash equilibrium}.

Another important result is that no producer $p \in t_\prods$ from any non-Byzantine coalition $t$ 
can avoid sending the expected $\mhash$ and $\mblock$ messages to consumers not from $t_\cons$.
The same applies to $\mreport$ messages sent by consumers to TO.
Therefore, for any $i \in t$, the expected utility of delaying messages to players not from $t$ is at most as high
as the utility of following the algorithm. Therefore, by the promptness principle, players never
delay messages between different coalitions. Concerning the messages exchanged between players from the same
coalition, we do not guarantee that players do not incur any communication delays. Though,
if these messages are mandatory to ensure that all players of the coalition are rewarded,
then, if there is any delay, it must be finite, otherwise, the expected utility is the same
as not sending these messages, i.e., at most 0.

\section*{Acknowledgements}
This work was partially supported by the FCT (INESC-ID multi annual
funding through the PIDDAC Program fund grant and by the project
PTDC/EIA-EIA/102212/2008).

\bibliographystyle{splncs}
\bibliography{bibfile}

\begin{thebibliography}{10}

\bibitem{boinc}
Anderson, D.:
\newblock Boinc: A system for public-resource computing and storage.
\newblock In: Proceedings of the 5th IEEE/ACM International Workshop on Grid
  Computing. GRID'04, Pittsburgh, PA, USA, IEEE (November 2004)  4--10

\bibitem{barb}
Aiyer, S., Alvisi, L., Clement, A., Dahlin, M., Martin, J.P., Porth, C.:
\newblock {BAR} fault tolerance for cooperative services.
\newblock In: Proceedings of the 20th ACM Symposium on Operating Systems
  Principles. SOSP'05, Brighton, United Kingdom, ACM (October 2005)  45--58

\bibitem{Vilaca:11}
Vila\c{c}a, X., Leit{\~a}o, J., Correia, M., Rodrigues, L.:
\newblock N-party {BAR} transfer.
\newblock In: Proceedings of the 15th International Conference On Principles Of
  Distributed Systems (to appear). OPODIS'11, Toulouse, France (December 2011)

\bibitem{eliaz}
Eliaz, K.:
\newblock Fault-tolerant implementation.
\newblock Review of Economic Studies \textbf{69}(3) (August 2002)  589--610

\bibitem{moscibroda}
Moscibroda, T., Schmid, S., Wattenhofer, R.:
\newblock On the topologies formed by selfish peers.
\newblock In: Proceedings of the 25th Annual ACM SIGACT-SIGOPS Symposium on
  Principles of Distributed Computing. PODC'06, Denver, CO, USA, ACM (July
  2006)  133--142

\bibitem{Wong:11}
Wong, E.L., Clement, A., Levy, I., Alvisi, L., Dahlin, M.:
\newblock Regret freedom isn't free.
\newblock In: Proceedings of the 15th International Conference On Principles Of
  Distributed Systems (to appear). OPODIS'11, Toulouse, France (December 2011)

\bibitem{Aumann:59}
Aumann, R.J.:
\newblock {Acceptable points in General Cooperative \$n\$-person Games}.
\newblock In: Contributions to the Theory of Games {IV}. Number~40 in Annals of
  Mathematics Studies.
\newblock Princeton University Press, Princeton (1959)  287--324

\bibitem{Bernheim:87}
Bernheim, B., Peleg, B., Whinston, M.:
\newblock Coalition-proof nash equilibria i. concepts.
\newblock Journal of Economic Theory \textbf{42}(1) (June 1987)  1--12

\bibitem{Moreno:96}
Moreno, D., Wooders, J.:
\newblock Coalition-proof equilibrium.
\newblock Games and Economic Behavior \textbf{17}(1) (November 1996)  80--112

\bibitem{abraham}
Abraham, I., Dolev, D., Gonen, R., Halpern, J.:
\newblock Distributed computing meets game theory: robust mechanisms for
  rational secret sharing and multiparty computation.
\newblock In: Proceedings of the 25th Annual ACM SIGACT-SIGOPS Symposium on
  Principles of Distributed Computing. PODC'06, Denver, CO, USA, ACM (July
  2006)  53--62

\bibitem{bar-theory}
Clement, A., Napper, J., Li, H., Martin, J.P., Alvisi, L., Dahlin, M.:
\newblock Theory of {BAR} games.
\newblock In: Proceedings of the 26th Annual ACM SIGACT-SIGOPS Symposium on
  Principles of Distributed Computing. PODC'07, Portland, OR, USA, ACM (August
  2007)  358--359

\bibitem{bargossip}
Li, H., Clement, A., Wong, E., Napper, J., Roy, I., Alvisi, L., Dahlin, M.:
\newblock {BAR} gossip.
\newblock In: Proceedings of the 7th USENIX Symposium on Operating Systems
  Design and Implementation. OSDI'06, Seattle, WA, USA, USENIX Association
  (November 2006)  191--204

\bibitem{flightpath}
Li, H., Clement, A., Marchetti, M., Kapritsos, M., Robison, L., Alvisi, L.,
  Dahlin, M.:
\newblock Flightpath: Obedience vs choice in cooperative services.
\newblock In: Proceedings of the 8th USENIX Symposium on Operating Systems
  Design and Implementation. OSDI'08, San Diego, CA, USA, USENIX Association
  (December 2008)  355--368

\bibitem{firespam}
Mokhtar, S., Pace, A., {Qu{\'e}ma}, V.:
\newblock {FireSpam}: Spam resilient gossiping in the {BAR} model.
\newblock In: Proceedings of the 29th IEEE International Symposium on Reliable
  Distributed Systems. SRDS'10, New Delhi, India, IEEE (October 2010)  225--234

\bibitem{bar-altruism}
Wong, E.L., Leners, J.B., Alvisi, L.:
\newblock It's on me! the benefit of altruism in {BAR} environment.
\newblock In: Proceedings of the 25th International Symposium on Distributed
  Computing. DISC'10, Cambridge, USA, Springer (September 2010)  406--420

\bibitem{Cachin:978-3-642-15259-7}
Cachin, C., Guerraoui, R., Rodrigues, L.:
\newblock Introduction to Reliable and Secure Distributed Programming. 2nd
  edition edn.
\newblock Springer-Verlag New York, Inc. (2011)

\end{thebibliography}

\end{document}